\documentclass[conference]{IEEEtran}

\usepackage[linesnumbered,ruled,vlined]{algorithm2e}
\usepackage[noend]{algorithmic}
\usepackage[geometry]{ifsym}
\usepackage{ifsym}
\usepackage{amssymb}
\usepackage{bbm}
\usepackage{amscd}
\usepackage{dsfont}
\usepackage{amsmath}
\usepackage{epsfig}
\usepackage{graphics}
\usepackage{psfrag}
\usepackage{rotating}
\usepackage{amsmath}
\usepackage{amsfonts}
\usepackage{url}
\usepackage{color}
\usepackage[caption=false,font=footnotesize]{subfig}
\usepackage{epstopdf}
\usepackage{amsthm}
\usepackage{pgfplots}
\usepackage{tikz}
\usetikzlibrary{arrows,shapes,decorations,backgrounds}

\newtheorem{theorem}{Theorem}
\newtheorem{lemma}[theorem]{Lemma}

\newtheorem{definition}[theorem]{Definition}
\newtheorem{proposition}[theorem]{Proposition}

\newtheorem{rem}[theorem]{Remark}

\IEEEoverridecommandlockouts

\begin{document}
\title{Power Control in Massive MIMO with Dynamic User Population}

\author{\IEEEauthorblockN{ 
  Mohamad Assaad, Salah Eddine Hajri, Thomas Bonald, and Anthony Ephremides
\thanks{M.Assaad and S. Hajri are with the Laboratoire des Signaux et Systemes (L2S, CNRS), CentraleSupelec, France. \{Mohamad.Assaad, Salaheddine.hajri\}@centralesupelec.fr.} \\ \thanks{Thomas Bonald is with the department of Computer Science and Networking of Telecom ParisTech, France. \{thomas.bonald@telecom-paristech.fr\} } \\ \thanks{A. Ephremides is with the Department of Electrical and Computer Engineering and Institute for Systems Research, University of Maryland, College Park, MD 20742. \{etony@umd.edu\}. }}}



\maketitle

\begin{abstract}
\boldmath  
This paper considers the problem of power control in Massive MIMO systems taking into account the pilot contamination issue and the arrivals and departures of users in the network. Contrary to most of existing work  in MIMO systems that focuses on the physical layer with fixed number of users, we consider in this work that the users arrive dynamically and leave the network once they are served. We provide a power control strategy, having a polynomial complexity, and prove that this policy stabilizes the network whenever possible. We then provide a distributed implementation of the power control policy requiring low information exchange between the BSs and show that it achieves the same stability region as the centralized policy.  

\end{abstract}


\section{Introduction}\label{Intro}

	Multiuser MIMO is one of the main technologies that has been adopted for  wireless
	networks. It enables to exploit the degrees of freedom in the spatial domain in order to serve users in the same frequency band and time. Coupled with a large antenna array at the Base Station, the resulting massive MIMO system enables a huge increase in the network spectral and energy efficiency. Massive MIMO has been designated as a key technology in the 5G  wireless networks.
	It  was first proposed in \cite{mm0}. The idea was to  mimic the large processing gain provided by  spread-spectrum in 3G networks.	This gain is imitated through the use of a  large number of base station antennas in massive MIMO systems. 
    The  large  excess of  transmit antennas allows to  considerably improve the network capacity	through excessive spatial dimensions \cite{mm1}. It also enables to  average  out  the  effect  of  fast  fading  and  provides  extremely  accurate beamforming which allows to direct the signal into small areas \cite{mm2}. Furthermore, the numerous degrees-of-freedom offered by massive MIMO result into significantly reducing the transmit power \cite{mm3}.
    Most of prior work on massive MIMO systems concentrated on the network physical layer with little interest in  the dynamic nature of traffic at the flow level, where a flow typically corresponds to a file transfer \cite{mm1}-\cite{mm4}. On the other hand, most papers on flow level models of wireless networks rely on simple models of the physical layer, neglecting the impact of pilot signals \cite{flow1}-\cite{assaad5}.

To the best of our knowledge, this is the first work that considers the impact of  dynamic population of users in massive MIMO systems. We provide a power control framework that takes into account the pilot contamination and stabilizes the network for any arrival rate of users that lies inside the stability region. 
The main difficulty of the problem lies in the fact that the network is dynamic with variable number of users that interfere with each other due to pilot contamination. The throughput of each user is thus not convex/concave, which complicates further the development of power control strategies in general. In this paper, we provide a convex power control framework, in which the utility function is not the bit rate, and show that it stabilizes the network whenever  possible.

The remaining of the paper is organized as follows. The system model is provided in section II. In section III, the power control strategy and the stability analysis are described. The distributed implementation and its stability analysis are given in Section IV. Numerical results are given in Section V and Section VI concludes the paper. 


\section{System Model}\label{model}

\subsection{Physical layer model and Pilot contamination}

We consider a multi cell MIMO  scenario consisting of $L$ cells. The BS has $M$ transmit antennas and the user terminals UTs have single antenna each. The notation UT$_{k,l_k}$ denotes the $k$-th UT present in cell $l_k$.
We consider a discrete-time block-fading channel model where the channel
remains constant during a time equal to a given coherence interval and then changes
independently from one block to the other. 

In Massive MIMO, and due to the high number of antennas at the base station, the channel state information (CSI) is usually acquired via reverse link pilots. For high number of users (especially in multi cell scenarios), the channel information for a given user is polluted by undesired channels of the other users using the common pilot sequence. 


During the uplink training phase, the $k^{th}$ user in cell $l$ transmits its own pilot sequence $s_{k,l} \in \mathbf{R}^{\tau \times 1}$ of length $\tau$. $\tau$ is thus proportional to the number of orthogonal reverse link pilots. At the base station $l$, the pilot data received is then,
\[\mathbf{z}_l=\sum_{k=1}^K\left(\sqrt{p_s}\mathbf{S}_{k,l}\right)\sqrt{g_{k,l}}\mathbf{h}_{k,l} + n_l\]
where $\mathbf{S}_{k,l}=s_{k,l}\otimes \mathbf{I}_M$, $p_s$ is the pilot power, $g_{k,l}$ is the channel power gain (inversely proportional to the path loss) between user $k$ and base station $l$, $n_l$ is the white Gaussian noise,   and $\mathbf{h}_{k,l}$ is the small scale fading channel with distribution $\mathcal{CN}(\mathbf{0},\mathbf{I}_M)$. $\mathbf{h}_{k,l}$ $\forall$ $k,l$ are assumed to be  i.i.d. Each base station $l$ estimates the channels of its own users using an MMSE estimation. Let $\hat{\mathbf{f}}_{k,l}$ is the MMSE estimation of the channel between user $k$ and BS $l$. Let $\mathbf{\hat{F}}_l$ be the matrix containing the estimated channels of the users in cell $l$. 

In the downlink, the base station pre-codes the signals $\mathbf{x}$ of its own users before transmission. Two pre-coding schemes are commonly used in Massive MIMO namely the Conjugate Beamforming (CB and the Zero Forcing (ZF). In this paper, we focus on CB however our stability analysis results hold for ZF.   In the CB, the pre-coding matrix is simply $\mathbf{W}_l=\mathbf{\hat{F}}^*_l$ where $\mathbf{\hat{F}}^*_l$ is the conjugate of $\mathbf{\hat{F}}_l$. The terminals in each cell receive, 
\[y_{k,l}=\sum_l \sqrt{g_{k,l}}\mathbf{h}^H_{k,l}\mathbf{W}_l \mathbf{d}_l+b_{k,l}\]
 where $\mathbf{d}_l$ contains the data of all users in cell $l$ and $b_{k,l}$ is the additive Gaussian noise. 

The downlink effective SINR for user $k$ in cell $l_k$ can then be obtained \cite{mm3}:

\begin{equation}
\gamma_{k,l_k}=\frac{\frac{\nu_k \tau \rho p_{k,l_k}g^2_{k,l_k}}{1+\tau \rho q_k}}{1+\sum_{j=1}^L \theta_{k,j}\sum_{i=1}^{K_j} p_{i,j_i}+\tau \rho \sum_{i \in \mathcal{I}_k \setminus k}\frac{\nu_{i,l_i} p_{i,l_i}g^2_{k,l_i}}{1+\tau \rho q_i}}
\end{equation}
where $l_j$ is the index of the base station that serves user $j$, $K_j$ is the number of users in cell $j$, $\mathcal{j}$ is the set of terminals using the same pilot sequence as terminal $j$, $\rho$ is the uplink pilot channel SNR, and $q_i=\sum_{j \in \mathcal{I}_i} g_{j,l_i}$ (i.e. the sum of large scale fading channel gain between the base station serving terminal $i$ and all terminals using the same pilot sequence as terminal $i$ including terminal $i$ itself). $\nu_{k,l}=M-1$ and $\theta_{k,j}= g_{k,j}$. 

By using the notation $G_{k,l_k}=\frac{\nu_k \tau \rho g^2_{k,l_k}}{1+\tau \rho q_k}$, the downlink effective SINR can be written as, 
\begin{equation}
\gamma_{k,l_k}=\frac{ p_{k,l_k}G_{k,l_k}}{1+\sum_{j=1}^L \theta_{k,j}\sum_{i=1}^{K_j} p_{i,j}+ \sum_{i \in \mathcal{I}_k \setminus k} p_{i,l_i}G_{k,l_i}}
\end{equation}

\subsection{Flow/User Arrival model}

All the existing work in massive MIMO assumes that the network is static in the sense that a fixed set of mobiles are always receiving data from the BSs. In practice, however, the users  arrive dynamically and once they are served they leave the network. To the best of our knowledge, this is the first work that considers dynamic population   of users in massive MIMO. 

For mathematical tractability, we consider finite possible locations of the users in the cell. The number of locations could be however high to cover most/all possible locations. In each location, it may exist or not a user that requires to be served. Let $X_{k,l}(t)$ be the number of users in location $k$ in cell $l$ at time $t$. Each user has one flow to be served. The words "user" and "flow" denote hence the same thing, where a flow typically corresponds to a file transfer. It is worth mentioning that the extension of our model to the case where each user has  multiple flows is straightforward. In each location, the flows arrive  according to a Poisson process with rate $\lambda_{k,l}$. The data volume to be transmitted to each user has an exponential distribution. Note that the physical layer model described above stays valid under this assumption. The SINR $\gamma_{k,l_k}$ represents the SINR of the user in location $k$ in cell $l_k$.

We assume a separation of time-scales between the flow model and the physical layer. The presence of interference and pilot contamination induces a non convexity in the problem that complicates the power control and stability analysis. The index time $t$ refers to the flow level time. Between times $t$ and $t+1$ there are multiple physical layer time slots.

We can provide the following definition of network stability.   

\begin{definition}
The network is called stable if $\lim_{T\rightarrow\infty}\sup\frac{1}{T}\sum_{t=0}^{T-1} \sum_{l=1}^L\sum_{k=1}^{K_l}\mathbb{E}(X_{k,l}(t))< \infty$
\end{definition}

Let $\Lambda = \left(\lambda_{1,1}, ...,\lambda_{k,i},...\right)$ be the set of all flow/user arrival rates. This rate is defined as the average number of users arriving in the network. The arrival of users means that those users become active and start a connection. 

\begin{definition}
The stability region is the set of all possible rates  of arriving users $\Lambda$ such that there exists a corresponding power control policy  that makes the network stable. 
\end{definition}

We are therefore interested in providing a power control strategy that makes the network stable.




\section{Power Allocation Framework} 
Let $\Lambda^{max}$ be the stability region under power control policies. 

The goal is then the following 

$$\textrm{Allocate } \ \mathbf{p}$$
s.t. 
$$ \forall \mathbf{\lambda} \in \Lambda^{max} \ \textrm{the network is stable}$$

It is worth mentioning that the value of the average user arrival rate $\mathbf{\lambda}$ is a priori not known and the allocation policy must stabilize the network without the knowledge of $\mathbf{\lambda}$ (since the average number of users arriving in the network is unknown).  In this section, we provide a power control policy that stabilizes the network for any $\mathbf{\lambda} \in \Lambda^{max}$. It is worth noticing that our main contribution in this section lies in proving that the considered power allocation policy stabilizes the network $\forall \mathbf{\lambda} \in \Lambda^{max}$.

\subsection{Power Control Framework with Polynomial Complexity}

To do so, we consider the following framework, 
\begin{eqnarray}
\max_{\mathbf{p}} \sum_{l=1}^L \sum_{k=1}^{K_l}X_{k,l} log\left(\gamma_{k,l}(t)\right) \label{frame1}
\end{eqnarray}
s.t.
\begin{equation}
\sum_{k=1}^{K_l}p_{k,l_k} \leq P_l \ \ \ \forall l \label{frame2}
\end{equation}
where $$\gamma_{k,l}(t)= \frac{ p_{k,l_k}G_{k,l_k}}{1+\sum_{j=1}^L \theta_{k,j}\sum_{i=1}^{K_j} p_{i,j_i}+ \sum_{i \in \mathcal{I}_k \setminus k} p_{i,l_i}G_{k,l_i}}$$ To simplify the notation, we will denote the power of flows in location $k$ in cell $l_k$ by $p_{k,l}$.

By using the following variable change $\tilde{p}_{k,l}=log\left(p_{k,j}\right)$, we get,
\begin{equation} 
\min_{\mathbf{\tilde{p}}} - \sum_{l=1}^L \sum_{k=1}^{K_l} X_{k,l} U_{k,l}(\mathbf{\tilde{p}}) \label{framePower1}
\end{equation}
s.t.
\begin{equation}
\sum_{k=1}^{K_l}e^{\tilde{p}_{k,l}} \leq P_l \ \ \ \forall l \label{frameConstraint1}
\end{equation}
where 
\[U_{k,l}(\mathbf{\tilde{p}}) = log\left(\frac{e^{\tilde{p}_{k,l}}G_{k,l}}{1+\sum_{\substack{j=1}}^L \theta_{k,j}\sum_{\substack{i=1}}^{K_j} e^{\tilde{p}_{i,j}} + \sum_{\substack{i \in \mathcal{I}_k \setminus k}} e^{\tilde{p}_{i,l}} G_{k,l_i}}\right)
\]

\begin{rem}
In the aforementioned problem we do NOT make the high SINR approximation $log\left(1+\gamma_{k,l}(t)\right) \approx log\left(\gamma_{k,l}(t)\right)$. We have just defined a utility function proportional to $log\left(\gamma_{k,l}(t)\right)$. The rate expression is still given by $log\left(1+\gamma_{k,l}(t)\right)$. Our main contribution in this paper is to show that without making such high SINR approximation, solving the problem  (\ref{frame1})-(\ref{frame2}),  with bit rate $log\left(1+\gamma_{k,l}(t)\right)$, ensures the stability of the network at the flow level under dynamic arrivals and departures of users. One can see also that the considered power control is different from the proportional fairness policy. 
\end{rem}

\begin{proposition}
The above optimization problem is convex.
\end{proposition}
\begin{proof} The proof follows from \cite{chiang} and is given in the appendix for completeness.  
\end{proof}

The optimal solution of the aforementioned problem can then be obtained simply using the Lagrangian technique and KKT conditions. The Lagrangian is given as,

\begin{eqnarray}
L(\mathbf{p},\mathbf{\lambda})= -\sum_{l=1}^L \sum_{k=1}^{K_l} X_{k,l} \biggl[ log\left(e^{\tilde{p}_{k,l}}G_{k,l}\right) -\nonumber \\   log\left(1+\sum_{\substack{j=1}}^L \theta_{k,j}\sum_{\substack{i=1}}^{K_j} e^{\tilde{p}_{i,j}} + \sum_{\substack{i \in \mathcal{I}_k \setminus k}} e^{\tilde{p}_{i,l}} G_{k,l_i}\right) \biggr] \nonumber \\ + \sum_{l=1}^L \beta_l \left(\sum_{k=1}^{K_l}e^{\tilde{p}_{k,l}} - P_l \right)
\end{eqnarray}
One can check easily that slater condition is satisfied for the aforementioned problem. The problem can be solved with zero duality gap and the optimal solution can be obtained by solving the dual problem $\max_{\mathbf{\beta}} \min_{\mathbf{p}} L(\mathbf{p},\mathbf{\beta})$. The optimal power to be allocated to each user is then,
\begin{equation}
p^*_{k,l}= \frac{X_{k,l}}{\beta_l+ h_{k,l}(\mathbf{p}^*)}\label{optimum}
\end{equation}
where 
\[h_{k,l}(\mathbf{p}^*)=\sum_{l=1}^L\sum_{j}\frac{X_{j,l}(\theta_{j,l_k}+G_{j,l_k}\mathbf{1}_{k \in I_j \setminus j})}{1+\sum_{\substack{j=1}}^L \theta_{k,j}\sum_{\substack{i=1}}^{K_j} p^*_{i,j} + \sum_{\substack{i \in \mathcal{I}_j \setminus j}} p^*_{i,l} G_{k,l_i}}\]
The above power is obtained by taking $\frac{\partial  L(\mathbf{p},\mathbf{\beta})}{ \partial \tilde{p}_{k,l}}=0$. In order to find the optimal power, the above equation must be solved and the optimal Lagrange multipliers  $\beta_l$ must be determined. This can be done using the following algorithm. 

\begin{algorithm}[htb]
\caption{Optimal Power Control  Algorithm}
\label{alg:I}
\begin{enumerate}
\item Initialize: $\mathbf{\beta}$, $\mathbf{p}$, $\epsilon$ is a very small value;
\item Initialize $n=0$ and $i=0$ ($n$ and $i$ are the indices of two loops) 
\item
\begin{enumerate}
\item For given value of $\beta_{l}(i)$ $\forall$ $l$, repeat
 \[p_{k,l}(n+1)=\frac{X_{k,l}}{\beta_l(i) + h_{k,l}(\mathbf{p}(n))}\] $\forall$ $k,l$
\item until $p_{k,l}(n+1)=\frac{X_{k,l}}{\beta_l(i) + h_{k,l}(\mathbf{p}(n))}= p_{k,l}(n)$ $\forall$ $k,l$ (i.e. a fixed point is achieved); Call this fixed point $\mathbf{p}(\beta)$;
\end{enumerate}
\item Update $\beta_l(i+1)=\left(\beta_l(i)+\delta_i(\sum_{k=1}^{K_l}p_{k,l}(\beta) - P_l)\right)^+$ and set $i=i+1$
\item if $\beta_l(i)=0$ $\forall l$ then terminate. 
\item Else if $(\sum_{k=1}^{K_l}p_{k,l}(\mathbb{\beta}) - P_l) \leq \epsilon$ $\forall$ $l$ for which $\beta_l(i)>0$, then terminate. 
\item Else, return to Step 3.  
\end{enumerate}
\end{algorithm}
The algorithm consists of two loops. The outer loop updates the values of $\beta_l$ using the sub gradient method i.e. $\beta_l(i+1)=\left(\beta_l(i)+\delta_i(\sum_{k=1}^{K_l}p_{k,l}(\beta(i)) - P_l)\right)^+$.  For each value of $\beta_l$ (denoted by $\beta_l(i)$ (i.e. iteration $i$ of the outer loop), the inner loop updates the values of power $p_{k,l}$ using the fixed point equation $p_{k,l}(n+1)=\frac{X_{k,l}}{\beta_l(i) + h_{k,l}(\mathbf{p}(n))}$ until a fixed point is achieved. 
\subsubsection{Convergence of Algorithm (\ref{alg:I})}
One can see that the function $\frac{X_{k,l}}{\beta_l + h(\mathbf{p})}$ is a standard function (one can refer to \cite{yates}\cite{JP09} for more details): 
\begin{itemize}
\item It is positive 
\item Monotone: if $\mathbf{p}_1 \geq \mathbf{p}_2$ then $\frac{X_{k,l}}{\beta_l + h(\mathbf{p}_1)} \geq \frac{X_{k,l}}{\beta_l+ h(\mathbf{p}_2)}$ 
\item Scalable: for $\gamma >1$ then $\frac{X_{k,l}}{\beta_l + h(\gamma \mathbf{p})} < \gamma \frac{X_{k,l}}{\beta_l + h(\mathbf{p})}$
\end{itemize}
Consequently, for each value of $\beta_l$ ($\forall$ $l$), the algorithm 
\[p_{k,l}(n+1)=\frac{X_{k,l}}{\beta_l + h_{k,l}(\mathbf{p}(n))}\]
converges to the fixed point $p_{k,l}^*=\frac{X_{k,l}}{\beta_l + h_{k,l}(\mathbf{p}^*)}$ \cite{yates}\cite{JP09}. This shows the convergence of the inner loop. Concerning the outer loop, the update $\beta_l(i+1)=\left(\beta_l(i)+\delta_i(\sum_{k=1}^{K_l}p_{k,l}(\beta) - P_l)\right)^+$ is nothing but the sub-gradient update method which is, due to the convexity of our optimization framework,  ensured to converge to the optimal value of $\beta_l$ (say $\beta^*_l$) for vanishing step size  $\delta_i$ and to within a close interval around $\beta^*_l$ for fixed step size $\delta_i=\delta$. One can refer to \cite{boyd} for more details on the convergence of sub-gradient descent method. 

\subsection{Stability analysis of the network}
In this section, we will show that the aforementioned power allocation policy stabilizes the network $\forall$  $\mathbf{\lambda}$ $\in$ $\Lambda^{max}$.  In other words, as far as the network stability is concerned there is no other allocation policy that can outperform the aforementioned power control policy.   
For that, we use the Fluid limit machinery to prove the stability. The main idea is to define a new process, say $\mathbf{Y}(t)$ by scaling or compressing time and accordingly scaling down the magnitude of the process $\mathbf{X}$. The process $\mathbf{Y}(t)$ can be seen as a deterministic fluid process driven by a fluid arrival process with constant rate.  In order to show that the network is stable under the power allocation policy (\ref{optimum}), it is sufficient to show that a Lyapunov function of the  fluid limit trajectory has a negative drift \cite{bonald2001,andrews2004,fluid79}.

We consider that the arrivals of the users follow a poisson distribution. Each user has a document, of size following a general distribution with mean 1, to be served. However, this assumption can be relaxed to renewal arrival processes and general document size distribution with mean $m_s$. We have that $\mathbf{X}=(X_{k,l})$, $\forall$ $k,l$, is a Markov process.  At each time $t$, the number of flows/users evolves as follows, 
\[ X_{k,l }\rightarrow X_{k,l}+1 \ \ \ \textrm{at rate} \ \ \ \lambda_{k,l} \]

\[ X_{k,l} \rightarrow X_{k,l}-1 \ \ \  \textrm{at rate} \ \ \ R_{k,l} \]
where $R_{k,l}$ is the bit rate allocated between $t$ and $t+1$. It is worth mentioning that for given values of $X_{k,l}$, the allocated power is the same. Recall that the time between $t$ and $t+1$ corresponds to multiple physical layer timeslots. For any time $u \in [t ; t+1]$, the allocated rate at the physical layer is 
\[ R_{k,l}= log(1+\frac{p_{k,l}G_{k,l}}{1+\sum_{j=1}^L \theta_{k,j}\sum_{i=1}^{K_j} p_{i,j}+ \sum_{i \in \mathcal{I}_k \setminus k} p_{i,l}G_{k,l_i}})\]

\begin{theorem}
The network is stable under power control (\ref{frame1})-(\ref{frame2}) $\forall$  $\mathbf{\lambda}$ $\in$ $\Lambda^{max}$. \label{stabilityI}
\end{theorem}

\begin{proof}
The proof is based on studying the fluid system obtained by framework (\ref{frame1})-(\ref{frame2}). 
The fluid system is obtained when the number of flows tends to $\infty$, 
\[Y_{k,l}(t)= \underset{N \rightarrow \infty}\lim \frac{X_{k,l}(Nt)}{N}\] with $\sum_{l=1}^L X_{k,l}(0)=N$. By the strong law of large number, the evolution of the process $\mathbf{Y}(t)$ is given by,
\begin{equation}
\frac{d}{dt}Y_{k,l}(t)=\lambda_{k,l}-R_{k,l} \label{fluid}
\end{equation}
One can note that framework (\ref{frame1})-(\ref{frame2}) is equivalent to 
\begin{equation}
\max_{\mathbf{R}}\prod_{k,l} \left(e^{R_{k,l}}-1\right)^{Y_{k,l}} \label{frame11}
\end{equation}
s.t. 
\begin{equation}
\sum_{k=1}^{K_l}f_l(\mathbf{R}) \leq P_l \ \ \ \forall l \label{frame22}
\end{equation}
where $f_l(\mathbf{R})$ is the equivalent of the max power constraint with respect to $\mathbf{R}$. $\mathbf{R}$ is the vector containing all the rates $R_{k,l}$.  (\ref{frame22}) is a feasibility constraint that determines (in addition to the interference) the rate region of the system. 

The above optimization problem is equivalent to,
\begin{equation}
\max_{\mathbf{R}}Log\left(\prod_{k,l} \left(e^{R_{k,l}}-1\right)^{Y_{k,l}}\right) \label{frame13}
\end{equation}
s.t. 
\begin{equation}
\sum_{k=1}^{K_l}f_l(\mathbf{R}) \leq P_l \ \ \ \forall l \label{frame23}
\end{equation}
The aforementioned objective function is equal to $\sum_{l=1}^L\sum_{k=1}^{K_l}Y_{k,l} Log\left(e^{R_{k,l}}-1\right)$ which is strictly concave with respect to $R_{k,l}$. This concavity implies the following. Let $\mathbf{p^*}$ the vector of optimal powers $p^*_{k,l}$ obtained in (\ref{optimum}). The corresponding rate $R^*_{k,l}= Log\left(1+\frac{p^*_{k,l}G_{k,l}}{1+\sum_{j=1}^L \theta_{k,j}\sum_{i=1}^{K_j} p^*_{i,j}+ \sum_{i \in \mathcal{I}_k \setminus k} p^*_{i,l}G_{k,l_i}}\right)$ is the optimal solution of framework (\ref{frame13}-\ref{frame23}). Therefore, by concavity of $Log\left(e^{R_{k,l}}-1\right)$, we have for all rates $\mu_{k,l}$:
\[ Log\left(e^{R^*_{k,l}}-1\right) \leq  Log\left(e^{\mu_{k,l}}-1\right) + \frac{e^{\mu_{k,l}}}{e^{\mu_{k,l}}-1}(R^*_{k,l}-\mu_{k,l})   \]
This implies that for any rate vector $\mathbf{\mu}$ lying inside the rate region, defined by condition (\ref{frame23}) and the interference, we have  $\sum_{k,l} Y_{k,l} Log\left(e^{R^*_{k,l}}-1\right) \geq \sum_{k,l} Y_{k,l} Log\left(e^{\mu_{k,l}}-1\right)$ since $\mathbf{R^*}$ is the optimal solution to (\ref{frame13}-\ref{frame23}). Consequently, from the optimality of  $\mathbf{R^*}$ and the aforementioned concavity inequality we get
\begin{equation}
\sum_{k,l}Y_{k,l}  \frac{e^{\mu_{k,l}}}{e^{\mu_{k,l}}-1}(\mu_{k,l} - R^*_{k,l}) \leq 0 \label{concavity1}
\end{equation}
Consider now the following Lyapunov function $\mathcal{L}(\mathbf{Y}(t))=\sum_{k,l}\frac{1}{2}\frac{e^{(\mu_{k,l}+\epsilon)}}{e^{(\mu_{k,l}+\epsilon)}-1}Y_{k,l}^2(t)$. For any arrival rate $\mu_{k,l}$ lying strictly inside the rate region, $\epsilon$ is selected such that  $(\mu_{k,l}+\epsilon)$ is in the region (or on the boundary of the region). We have
\begin{equation}
\frac{d\mathcal{L}(\mathbf{Y}(t))}{dt}= \sum_{k,l}Y_{k,l}  \frac{e^{(\mu_{k,l}+\epsilon)}}{e^{(\mu_{k,l}+\epsilon)}-1}(\mu_{k,l} - R^*_{k,l}) 
\end{equation}
From (\ref{concavity1}), applied to $(\mu_{k,l}+\epsilon)$, we know that $\sum_{k,l}Y_{k,l}  \frac{e^{(\mu_{k,l}+\epsilon)}}{e^{(\mu_{k,l}+\epsilon)}-1}((\mu_{k,l}+\epsilon) - R^*_{k,l}) \leq 0$. This implies that 
\[ \frac{d\mathcal{L}(\mathbf{Y}(t))}{dt} \leq -\epsilon \sum_{k,l}Y_{k,l}  \frac{e^{(\mu_{k,l}+\epsilon)}}{e^{(\mu_{k,l}+\epsilon)}-1} \]
We conclude that If $\mathcal{L}(\mathbf{Y}(t)) \geq \delta_1$ then the allocation policy in (\ref{optimum}) decreases the Lyapunov function by an amount proportional to $\mathbf{Y}(t)$. This implies, from \cite{bonald2001, andrews2004,fluid79} that the system is stable.

\end{proof}

\section{Distributed implementation with Low Information Exchange}
In practice, each base station allocates the power to the users without exchanging too much information with other base stations. In this section, we provide a distributed power control policy that requires very few signaling overhead.  

At a first look, one can use Algorithm 2.

\begin{algorithm}[htb]
\caption{Distributed Power Control  Algorithm}
\label{alg:II}
\begin{enumerate}
\item Initialize: $\mathbf{\beta}$, $\mathbf{p}$, $\epsilon$ is a very small value;
\item For given value of $\beta_{l}$ $\forall$ $l$: 
\begin{enumerate}
\item Repeat:
\item Each user $j$ calculates/estimates the expressions  $q_{k,l}=1+\sum_{\substack{j=1}}^L \theta_{k,j}\sum_{\substack{i=1}}^{K_j} p^*_{i,j} + \sum_{\substack{i \in \mathcal{I}_j \setminus j}} p^*_{i,l} G_{k,l_i}$  and feeds back it to the base station
\item The BSs exchange all $q_{k,l}$ between each other. 
\item Each BS calculates the function $h_{k,l}(\mathbf{p})=\sum_{l=1}^L\sum_{j}\frac{\theta_{j,l_k}+G_{j,l_k}\mathbf{1}_{k \in I_j \setminus j}}{1+\sum_{\substack{j=1}}^L \theta_{k,j}\sum_{\substack{i=1}}^{K_j} p^*_{i,j} + \sum_{\substack{i \in \mathcal{I}_j \setminus j}} p^*_{i,l} G_{k,l_i}}=\sum_{l=1}^L\sum_{j}\frac{\theta_{j,l_k}+G_{j,l_k}\mathbf{1}_{k \in I_j \setminus j}}{q_{j,l}}$; allocates and transmits the power $p_{k,l}(n+1)=\frac{X_{k,l}}{\beta_l + h_{k,l}(\mathbf{p}(n))}$ $\forall$ $k$
\item Until a fixed point is achieved  $p_{k,l}=\frac{X_{k,l}}{\beta_l + h_{k,l}(\mathbf{p})}$ $\forall$ $k$
\end{enumerate}
\item Each BS does the following test: 
\item If $(\sum_{k=1}^{K_l}p_{k,l}(\mathbb{\beta}) - P_l) \leq \epsilon$ $\forall$ $l$ Stop : the power is obtained 
\item Else each BS updates $\beta_l(i+1)=\left(\beta_l(i)+\delta_i(\sum_{k=1}^{K_l}p_{k,l}(\beta) - P_l)\right)^+$ and goes to step 2
\end{enumerate}
\end{algorithm}
The above algorithm will converge to the optimal solution. However, it suffers from a main weakness: a huge number of information must be exchanged between the BSs before the convergence. In this section, we therefore take advantage of the particularity of our problem and propose another distributed algorithm that requires small information exchange between the BSs. We will show then that our algorithm can stabilize the network for any user arrival rate  $\mathbf{\lambda}$ $\in$ $\Lambda^{max}$.


\subsection{Distributed Algorithm with Low Information Exchange}
The main idea of our algorithm is as follows. Recall from the system model that $(k,l)$ means a given position in cell $l$. Therefore $G_{k,l}$ and $\theta_{k,l}$ are well know (function of the path loss for a given position). These values need to be exchanged once at the beginning of connections (or at least once every few seconds). If the BSs exchange the values of their $X_{k,l}$, each BS will therefore have all the required parameters to build locally the whole optimization framework:
\begin{eqnarray}
\max_{\mathbf{p}} \sum_{l=1}^L \sum_{k=1}^{K_l}X_{k,l} log\left(\gamma_{k,l}(t)\right)
\end{eqnarray}
s.t.
\begin{equation}
\sum_{k=1}^{K_l}p_{k,l_k} \leq P_l \ \ \ \forall l
\end{equation}
where $$\gamma_{k,l}(t)= \frac{ p_{k,l_k}G_{k,l_k}}{1+\sum_{j=1}^L \theta_{k,j}\sum_{i=1}^{K_j} p_{i,j_i}+ \sum_{i \in \mathcal{I}_k \setminus k} p_{i,l_i}G_{k,l_i}}$$ 
Notice that the number of users $X_{k,l}$ changes slowly i.e. once each multiple time slots (as explained earlier in the paper). In this section, we go further and reduce more the exchange of information between the BSs. We consider that each  BS quantizes its values of $X_{k,l}$ and exchanges the quantized values with the other BS only once eaxh $D$ slots. In other words, each BS $j$ knows an outdated noisy version of the values of $X_{k,l}$ of other BS $k$. We denote this quantized outdated value by $\hat{X}_{k,l}$. The algorithm is then as follows: 

\begin{algorithm}[htb]
\caption{Distributed Power Control  Algorithm with small Information Exchange}
\label{alg:III}
\begin{itemize}
\item Each BS $l$ will use the quantized outdated value $\hat{X}_{i,j}$ instead of $X_{i,j}$ $\forall$ $i,j$ (including $j=l$ i.e. even for its own values) 
\item Using $\hat{X}_{i,j}$ $\forall$ $i,j$, BS $l$ performs Algorithm (\ref{alg:I}) and obtain the whole vector of power $p_{k,l}, ..., p_{i,j},...$ $\forall$ $i,j$. It drops of course  the values  $p_{i,j}$ $\forall$ $j\neq l$ and uses  only its own values of $p_{k,l}$ for transmission. 
\item Each BS has then its own vector $p_{k,l}$ and transmission can be performed.  
\end{itemize}
\end{algorithm}
Recall that in Algorithm (\ref{alg:I}), the update of any power value $p_{k,l}$ requires the knowledge of all the other values $p_{i,j}$ which explains the need for each BS to obtain the whole power vector $p_{k,l}, ..., p_{i,j},...$ $\forall$ $i,j$. Of course, the power values $p_{k,l}, ..., p_{i,j},...$ $\forall$ $i,j$ obtained by the BSs are identical due to the fact that the same $\hat{X}_{i,j}$ are used all BSs.  It is obvious that the aforementioned algorithm achieves the optimal solution of the following problem,

\begin{eqnarray}
\max_{\mathbf{p}} \sum_{l=1}^L \sum_{k=1}^{K_l}\hat{X}_{k,l} log\left(\gamma_{k,l}(t)\right)
\end{eqnarray}
s.t.
\begin{equation}
\sum_{k=1}^{K_l}p_{k,l_k} \leq P_l \ \ \ \forall l
\end{equation}

\subsection{Stability Analysis of the Distributed Algorithm}
We will show in the next theorem that our distributed algorithm stabilizes the network for any average number of arriving users $\mathbf{\lambda}$ $\in$ $\Lambda^{max}$. In other words, in terms of stability of the network, our distributed algorithm achieves the same performance as the optimal centralized algorithm.

\begin{theorem}
Algorithm (\ref{alg:III}) stabilizes the network $\forall$  $\mathbf{\lambda}$ $\in$ $\Lambda^{max}$.
\end{theorem}

\begin{proof} 
The proof is based on studying the fluid system.  
Let \[Y_{k,l}(t)= \underset{N \rightarrow \infty}\lim \frac{X_{k,l}(Nt)}{N}\] with $\sum_{l=1}^L X_{k,l}(0)=N$ and  \[\hat{Y}_{k,l}(t)= \underset{N \rightarrow \infty}\lim \frac{\hat{X}_{k,l}(Nt)}{N}\] with $\sum_{l=1}^L \hat{X}_{k,l}(0)=N$. 

The distributed algorithm (\ref{alg:III}) solves the optimization problem  
\begin{eqnarray}
\max_{\mathbf{p}} \sum_{l=1}^L \sum_{k=1}^{K_l}\hat{Y}_{k,l} log\left(\gamma_{k,l}(t)\right) \label{Dist}
\end{eqnarray}
s.t.
\begin{equation}
\sum_{k=1}^{K_l}p_{k,l_k} \leq P_l \ \ \ \forall l \label{DistConst}
\end{equation}

Let $\mathbf{p^*}$ be the vector of optimal powers $p^*_{k,l}$ obtained by (\ref{alg:III}). Its corresponding rate $R^*_{k,l}= Log\left(1+\frac{p^*_{k,l}G_{k,l}}{1+\sum_{j=1}^L \theta_{k,j}\sum_{i=1}^{K_j} p^*_{i,j}+ \sum_{i \in \mathcal{I}_k \setminus k} p^*_{i,l}G_{k,l_i}}\right)$. In a similar way as in the proof of Theorem (\ref{stabilityI}), we can show that for any rate vector $\mathbf{\mu}$ lying inside the rate region, we have
\begin{equation}
\sum_{k,l}\hat{Y}_{k,l}  \frac{e^{\mu_{k,l}}}{e^{\mu_{k,l}}-1}(\mu_{k,l} - R^*_{k,l}) \leq 0 \label{concavity2}
\end{equation}

The next step in the proof is to find upper and lower bounds of  the difference between $Y_{k,l}(t)$ and $\hat{Y}_{k,l}$. Let $R^{max}$ be the maximum transmission rate for all users, i.e. $R^{max}=\sup_{k,l,t}R_{k,l}(t)$, and $\lambda^{max}$ the maximum users' arrival rate i.e. $\lambda^{max}=\sup_{k,l}\lambda_{k,l}$. Recall that a quantized version of $X_{k,l}(t)$ is exchanged once each D slots. One can then see that 
\[\hat{X}_{k,l}(t)-D R^{max} - E_Q \leq X_{k,l}(t) \leq \hat{X}_{k,l}(t)+D\lambda^{max} +E_Q\]
where $E_Q$ is the max quantization error. The bounds are determined as follows. At time $t-D$ the real value of $X_{k,l}(t-D)$ is lower bound by $\hat{X}_{k,l}(t)-E_Q$. Then a lower bound for $X_{k,l}(t)$ at time $t$ can be obtained by assuming that no arrivals arise during $D$ slots and departures are at maximum rate $R^{max}$.  The upper bound can obtained by assuming that no departure arises during $D$ slots and arrivals are at max rate $\lambda^{max}$. By using fluid limit, we can bound $Y_{k,l}(t)$ as follows
\[\hat{Y}_{k,l}(t)-D R^{max} - E_Q \leq Y_{k,l}(t) \leq \hat{Y}_{k,l}(t)+D\lambda^{max} +E_Q\]
The above inequality implies 
\[|Y_{k,l}(t)-\hat{Y}_{k,l}(t)| \leq B\]
where $B=\max \left(D\lambda^{max} +E_Q, D R^{max} + E_Q\right)$. This implies that $\|Y_{k,l}(t)-\hat{Y}_{k,l}(t)\|$ is upper bounded by a constant independent of $Y_{k,l}(t)-$. 

Consider now the following Lyapunov function $\mathcal{L}(\mathbf{Y}(t))=\sum_{k,l}\frac{1}{2}\frac{e^{(\mu_{k,l}+\epsilon)}}{e^{(\mu_{k,l}+\epsilon)}-1}\left(Y_{k,l}(t)\right)^2$. Recall that by the strong law of large number, the evolution of the process $\mathbf{Y}(t)$ is given by,
\begin{equation}
\frac{d}{dt}Y_{k,l}(t)=\lambda_{k,l}-R_{k,l} \label{fluid}
\end{equation}
For any arrival rate $\mu_{k,l}$ lying strictly inside the rate region, $\epsilon$ is selected such that  $(\mu_{k,l}+\epsilon)$ is in the region (or on the boundary of the region). We have
\begin{eqnarray}
\frac{d\mathcal{L}(\mathbf{Y}(t))}{dt} &= & \sum_{k,l}\left(Y_{k,l}\right)  \frac{e^{(\mu_{k,l}+\epsilon)}}{e^{(\mu_{k,l}+\epsilon)}-1}(\mu_{k,l} - R^*_{k,l}) \nonumber \\ &=&  \sum_{k,l}\left(Y_{k,l}-\hat{Y}_{k,l}\right)  \frac{e^{(\mu_{k,l}+\epsilon)}}{e^{(\mu_{k,l}+\epsilon)}-1}(\mu_{k,l} - R^*_{k,l})  \nonumber \\ &+&  \sum_{k,l}\left(\hat{Y}_{k,l}\right)  \frac{e^{(\mu_{k,l}+\epsilon)}}{e^{(\mu_{k,l}+\epsilon)}-1}(\mu_{k,l} - R^*_{k,l}) \nonumber \\ &\leq &  \sum_{k,l} B  \frac{e^{(\mu_{k,l}+\epsilon)}}{e^{(\mu_{k,l}+\epsilon)}-1}|\mu_{k,l} - R^*_{k,l}|  \nonumber \\ &+&  \sum_{k,l}\left(\hat{Y}_{k,l}\right)  \frac{e^{(\mu_{k,l}+\epsilon)}}{e^{(\mu_{k,l}+\epsilon)}-1}(\mu_{k,l} - R^*_{k,l}) \nonumber \\ & \leq & C+ \sum_{k,l}\left(\hat{Y}_{k,l}\right)  \frac{e^{(\mu_{k,l}+\epsilon)}}{e^{(\mu_{k,l}+\epsilon)}-1}(\mu_{k,l} - R^*_{k,l}) \nonumber \\ \label{DLHat}
\end{eqnarray}
where $C=\max \left(  \sum_{k,l} B  \frac{e^{(\mu_{k,l}+\epsilon)}}{e^{(\mu_{k,l}+\epsilon)}-1}|\mu_{k,l} - R^*_{k,l}|  \right)$. One can notice that $C$ is independent of $Y_{k,l}(t)$ and $\hat{Y}_{k,l}(t)$ (Recall that $\mu_{k,l}$ and $R^*_{k,l}$ are both bounded) . 
From (\ref{concavity2}), applied to $(\mu_{k,l}+\epsilon)$, we know that $\sum_{k,l}\hat{Y}_{k,l}  \frac{e^{(\mu_{k,l}+\epsilon)}}{e^{(\mu_{k,l}+\epsilon)}-1}((\mu_{k,l}+\epsilon) - R^*_{k,l}) \leq 0$. Using the inequality (\ref{DLHat}), we get, 
\[ \frac{d\mathcal{L}(\mathbf{Y}(t))}{dt} \leq C -\epsilon \sum_{k,l}Y_{k,l}  \frac{e^{(\mu_{k,l}+\epsilon)}}{e^{(\mu_{k,l}+\epsilon)}-1} \]
The above inequality shows clearly that when $\mathcal{L}(\mathbf{Y}(t))$ is high, i.e. $\mathcal{L}(\mathbf{Y}(t)) \geq \Delta_1$ then the allocation policy in Algorithm (\ref{alg:III})  makes $ \frac{d\mathcal{L}(\mathbf{Y}(t))}{dt} <0$ i.e. the Lyapunov function decreases.  This implies, from \cite{andrews2004,fluid79} that the system is stable.

\end{proof}

\section{Numerical Results} 
	We provide numerical results illustrating the stability performance of the system. We consider an hexagonal cell network  with $2$ cells. Each cell has a  radius $1\;Km$ from center to vertex. Each base station is equipped with  $M=100$ antennas and the system bandwidth is $B=20\; MHz$.
	For the sake of tractability, we consider two possible locations for the users in each cell: i) the users are at the border of the cells or ii) the users are at  $r_0= 100 \;m$ from their serving BS.   In each location $k,l$, the  user flow arrive according to a Poisson process with average rate $\lambda_{k,l}$. The users will be receiving  packets of fixed size $1 Mbits$. 
	 For large scale fading coefficients we take into consideration only path-loss where $g_{k,l} = (r_{k,l})^{(-\sigma)}$, between the  $k^{th}$ user in the $l^{th}$ cell and its serving BS.  We took the path-loss exponent $\sigma=2.5$.  

	\begin{figure}[h!]
		\centering
		\includegraphics[width=7cm,height=5.5cm]{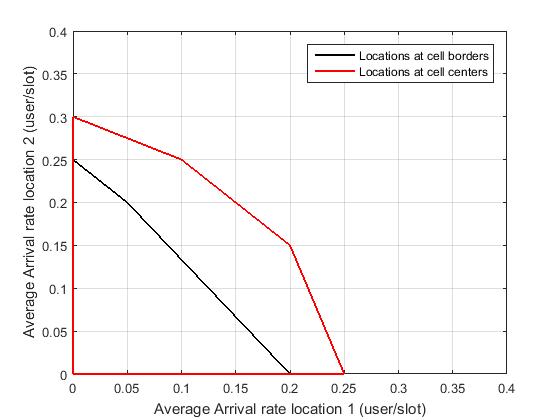}
		\caption{Stability Region}
	\end{figure}

\section{Conclusion}
In this paper, we have analyzed the stability of a multicellular network with massive MIMO and pilot contamination issue. Contrary to most of existing work, we consider that the number of users in the network is time varying. We have provided a simple power control policy and prove  that this policy stabilizes the network whenever possible. We have provided a distributed power allocation strategy that requires very low information exchange between the BSs and have shown that it achieves the same stability region as the centralized policy. 

\appendix 
\section*{Proof of convexity of the power control problem (\ref{framePower1}-\ref{frameConstraint1})}
\begin{proof} The objective function is 
\begin{eqnarray}
F= - \sum_{l=1}^L \sum_{k=1}^{K_l} X_{k,l} \biggl(\tilde{p}_{k,l}+log(G_{k,l})- \nonumber \\ log\left[1+\sum_{\substack{j=1}}^L \theta_{k,j}\sum_{\substack{i=1}}^{K_j} e^{\tilde{p}_{i,j}} + \sum_{\substack{i \in \mathcal{I}_k \setminus k}} e^{\tilde{p}_{i,l}} G_{k,l_i}\right]\biggr)
\end{eqnarray}
To show the convexity of the above objective function it is sufficient to prove that $\hat{F}_{k,l} = log\left[1+\sum_{\substack{j=1}}^L \theta_{k,j}\sum_{\substack{i=1}}^{K_j} e^{\tilde{p}_{i,j}} + \sum_{\substack{i \in \mathcal{I}_k \setminus k}} e^{\tilde{p}_{i,l}} G_{k,l_i}\right]$ is convex. The convexity of $F$ follows since the sum of convex and affine functions is convex. The convexity of $\hat{F}_{k,l} $ can be proven easily by showing that the Hessian $H_{k,l}$ is positive definite.

\[H_{k,l}=X_{k,l}\frac{1}{d^2_{k,l}}\left(d_{k,l}diag(\mathbf{b}^{k,l})-\mathbf{b}_{k,l}\mathbf{b}_{k,l}^T\right)\]
where $d_{k,l}=1+\sum_{\substack{j=1}}^L \theta_{k,j}\sum_{\substack{i=1}}^{K_j} e^{\tilde{p}_{i,j}} + \sum_{\substack{i \in \mathcal{I}_k \setminus k}} e^{\tilde{p}_{i,l}} G_{k,l_i}$ and $\mathbf{b}_{k,l}$ is a vector of length $L\times K$ where $b_{k,l}(i,j)=e^{\tilde{p}_{i,j}}\left(\theta_{k,j}+\mathbf{1}_{\{i\in  \mathcal{I}_k \setminus k\}}G_{k,j}\right)$. Notice that the index $(i,j)$ in  $b_{k,l}(i,j)$ means user $i$ in cell $j$ (i.e. $j=l_i$). We then can show that $\forall$ $\mathbf{u}$ $\in$ $\mathbb{R}^{K\times L}$
\begin{eqnarray}
\mathbf{u}^T H_{k,l}\mathbf{u}= \nonumber \\ \frac{X_{k,l}\left(d_{k,l}(\sum_{(i,j)}u^2_{i,j}b_{k,l}(i,j))-(\sum_{i,j}u_{i,j}b_{k,l}(i,j))^2\right)}{d^2_{k,l}} \nonumber \\ > 0
\end{eqnarray}
where the positivity follows from the Cauchy-Schwarz inequality (recall that $d_{k,l}$ can be written as $1+\sum_{i,j}b_{k,l}(i,j)$). One can refer to \cite{chiang} for more details. This completes the proof. 
\end{proof}





\end{document}